\documentclass{amsart}%
\usepackage{amsfonts}
\usepackage{amsmath}
\usepackage{amssymb}
\usepackage{graphicx}%
\providecommand{\U}[1]{\protect\rule{.1in}{.1in}}
\newtheorem{theorem}{Theorem}
\theoremstyle{plain}
\newtheorem{acknowledgement}{Acknowledgement}

\newtheorem{corollary}{Corollary}

\numberwithin{equation}{section}
\begin{document}
\title{The Amplified Quantum Fourier Transform:\\solving the local period problem}
\author{David J. Cornwell}
\curraddr{David J. Cornwell (PhD Student)\\
Department of Mathematics\\
University of Maryland Baltimore County\\
1000 Hilltop Circle\\
Baltimore, MD 21250}
\email{David J. Cornwell: dave.cornwell@yahoo.com}
\date{September 30th, 2010, Accepted into QIP Journal August 2nd, 2012}
\subjclass[2000]{Primary 05C38, 15A15; Secondary 05A15, 15A18}
\keywords{Quantum Fourier Transform, Amplitude Amplification, Oracle, Period Finding,
Shor Algorithm, Grover Algorithm}

\begin{abstract}
This paper creates and analyses a new quantum algorithm called the Amplified
Quantum Fourier Transform (Amplified-QFT) for solving the following problem:

\textbf{The Local Period Problem: }Let $L=\{0,1...N-1\}$ be a set of $N$
labels and let $A$ be a subset of $M$ labels of period $P$, i.e. a subset of
the form $A=\{j:j=s+rP,r=0,1...M-1\}$ where $P\leq\sqrt{N\text{ }}$and $M<<N$,
and where $M$ is assumed known. Given an oracle $f:L\rightarrow\{0,1\}$ which
is $1$ on $A$ and $0$ elsewhere, find the local period $P.$ A separate
algorithm finds the offset $s$.

The first part of this paper defines the Amplified-QFT algorithm. The second
part of the paper summarizes the main results and compares the Amplified-QFT
algorithm against the Quantum Fourier Transform (QFT) and Quantum Hidden
Subgroup (QHS) algorithms when solving the local period problem. It is shown
that the Amplified-QFT algorithm is, on average, quadratically faster than
both the QFT and\ QHS algorithms. The third part of the paper provides the
detailed proofs of the main results, describes the method of recovering $P$
from an observation $y$ and describes the algorithm for finding the offset $s$.

\begin{acknowledgement}
I would like to thank my advisor, Professor Samuel J. Lomonaco of the CSEE
Department of the University of Maryland Baltimore County (UMBC) for his
guidance and help in the writing of this paper.

\end{acknowledgement}
\end{abstract}
\maketitle

\section{\textbf{Introduction}}

This paper creates and analyzes a new quantum algorithm called the Amplified
Quantum Fourier Transform (Amplified-QFT) for solving the following problem:

\textbf{The Local Period Problem: }Let $L=\{0,1...N-1\}$ be a set of $N$
labels and let $A$ be a subset of $M$ labels of period $P$, i.e. a subset of
the form $A=\{j:j=s+rP,r=0,1...M-1\}$ where $P\leq\sqrt{N\text{ }}$and $M<<N
$, and where $M$ is assumed known. Given an oracle $f:L\rightarrow\{0,1\}$
which is $1$ on $A$ and $0$ elsewhere, find the local period $P.$ A separate
algorithm finds the offset $s$.

The first part of this paper defines the Amplified-QFT algorithm. The second
part of the paper summarizes the main results and compares the Amplified-QFT
algorithm against the Quantum Fourier Transform (QFT) and Quantum Hidden
Subgroup (QHS) algorithms when solving the local period problem. It is shown
that the Amplified-QFT algorithm is, on average, quadratically faster than
both the QFT and\ QHS algorithms. The third part of the paper provides the
detailed proofs of the main results, describes the method of recovering $P$
from an observation $y$ and describes the algorithm for finding the offset $s
$.

\section{Background-Amplitude Amplification}

In ref[4] Lov Grover specified a quantum search algorithm that searched for a
single marked element $x0$ in an $N$ long list $L$. An oracle $f:L\rightarrow
\{0,1\}$ is used to mark the element such that $f(x0)=1$ and $f$ is $0$
elsewhere. Grover's quantum algorithm finds the element with a work factor of
$O(\sqrt{N})$ whereas on a classical computer this would take $O(N)$, thereby
obtaining a quadratic speedup. Grover's algorithm can be summarized as follows:

a) Initialize the state to be the uniform superposition state $|\psi>=H|0>$
where $H$ is the Hadamard transform.

b) Reflect the current state about the plane orthogonal to the state $|x0>$ by
using the operator $(I-2|x0><x0|).$

c)\ Reflect the new state back around $|\psi>$ by using the operator $(2$
$|\psi>$ $<\psi|-I)$. This operator is a reflection about the average of the
amplitudes of the new state.

d) Repeat steps b) and c) $O(\sqrt{N})$ times until most of the probability is
on $|x0>.$

e) Measure the resulting state to obtain $x0.$

Also in ref[4], Grover suggested this algorithm could be extended to the case
of searching for an element in a subset $A$ of $M$ marked elements\ in an $N$
long list $L$. Once again an oracle $f:L\rightarrow\{0,1\}$ is used to mark
the elements of the subset $A.$Grover's algorithm solves this problem with a
work factor of $O(\sqrt{N/M})$. The elements of the set $A$ are sometimes
referred to as "good" and the elements not in $A$ are called "bad". Grover's
algorithm for this problem can be summarized as follows:

a) Initialize the state to be the uniform superposition state $|\psi>=H|0>$
where $H$ is the Hadamard transform.

b) Reflect the current state about the plane orthogonal to the state $|xgood>
$ by using the operator $(I-2|xgood><xgood|)$, where $xgood$ is the normalized
sum of the good states defined by the set $A.$This changes the sign of the
amplitudes of the good states defined by $A.$

c)\ Reflect the new state back around $|\psi>$ by using the operator $(2$
$|\psi>$ $<\psi|-I)$.

d) Repeat steps b) and c) $O(\sqrt{N/M})$ times until most of the probability
is on the set $A.$

e) Measure the resulting state to obtain an element in the set $A.$

Both versions of Grover's algorithm are also known as Amplitude Amplification
algorithms. which are generalized even further in ref [9]. The first part of
the Amplified-QFT algorithm consists of the second of these algorithms, except
for the final measurement step e).

\section{Background-Period Finding}

In ref[3], Peter Shor describes a quantum algorithm to solve the factorization
problem with exponential speed up over classical approaches. He translates the
factorization problem into one of finding the period of the function
$a^{x}ModN$ where $N$ is the number to be factored and $\gcd(a,N)=1$. The
period is found by making use of the QFT. Shor's factorization algorithm is
summarized below:

a) Find $Q:N^{2}\leq Q<2N^{2}$

b) Find $a:\gcd(a,N)=1$

c) Find the period of $a^{x}ModN$ using the QFT and using the $Qt\dot{h}$ root
of unity

\qquad- Form the superposition $\frac{1}{\sqrt{Q}}\sum|x>|a^{x}ModN>$

\qquad- Apply the QFT to the first register $|x>\rightarrow\sum\omega^{xy}|y>$

\qquad- Measure $y$

\qquad- Form the continued fraction expansion of $y/Q$ to find $d/P$

\qquad- If $|y/Q-d/P|<1/2N^{2}$ and $\gcd(d,P)=1$ then $\ P$ is recovered

d) If the period is not even start over

e) If $a^{P/2}+1=0ModN$ start over

f) Find $\gcd(a^{p/2}-1,N)$ to find the factor of $N.$

Step c) is the quantum part of Shor's factorization algorithm. We make use of
the QFT and continued fraction expansion method to recover the period $P$ in
the second part of the Amplified-QFT algorithm.

\section{The Amplified\ Quantum Fourier Transform Algorithm}

The Amplified-QFT algorithm solves the Local Period Problem:

\textbf{The Local Period Problem: }Let $L=\{0,1...N-1\}$ be a set of $N$
labels and let $A$ be a subset of $M$ labels of period $P$, i.e. a subset of
the form $A=\{j:j=s+rP,r=0,1...M-1\}$ where $P\leq\sqrt{N\text{ }}$and $M<<N
$, and where $M$ is assumed known. Given an oracle $f:L\rightarrow\{0,1\}$
which is $1$ on $A$ and $0$ elsewhere, find the local period $P.$ A separate
algorithm finds the offset $s$.

The Amplified-QFT algorithm consists of the following steps where steps a)
through d)\ are the Amplitude Amplification steps and steps e) through h) are
the period finding steps that use the QFT:

a) Initialize the state to be the uniform superposition state $|\psi>=H|0>$
where $H$ is the Hadamard transform.

b) Reflect the current state about the plane orthogonal to the state $|xgood>
$ by using the operator $(I-2|xgood><xgood|)$, where $xgood$ is the normalized
sum of the good states defined by the set $A.$This changes the sign of the
amplitudes of the good states defined by $A.$

c)\ Reflect the new state back around $|\psi>$ by using the operator $(2$
$|\psi>$ $<\psi|-I)$.

d) Repeat steps b) and c) $O(\sqrt{N/M})$ times until most of the probability
is on the set $A.$

e) Apply the QFT\ to the resulting state

f)\ Make a measurement $y$

g) Form the continued fraction expansion of $y/N$ to find $d/P$

h) If $|y/N-d/P|<1/2N^{2}$ and $\gcd(d,P)=1$ then $\ P$ is recovered

i) If $\gcd(d,P)\neq1$ repeat the algorithm starting at step a)

The Amplified-QFT algorithm produces the following states (See later sections
for the detailed analysis of the Amplified-QFT algorithm):

After applying steps b) and c) $k$ times where $k=\left\lfloor \frac{\pi
}{4\sin^{-1}(\sqrt{M/N})}\right\rfloor $ we arrive at the following state:%
\[
|\psi_{k}>=a_{k}\sum_{z\in A}|z>+b_{k}\sum_{z\notin A}|z>
\]

\noindent where%

\[
a_{k}=\frac{1}{\sqrt{M}}\sin(2k+1)\theta,b_{k}=\frac{1}{\sqrt{N-M}}%
\cos(2k+1)\theta
\]

\noindent are the appropriate amplitudes of the states and where
\[
\sin\theta=\sqrt{M/N},\cos\theta=\sqrt{1-M/N}%
\]

\noindent.

The QFT at step e) performs the following action%

\[
|z>\rightarrow\frac{1}{\sqrt{N}}\sum_{y=0}^{N-1}e^{-2\pi izy/N}|y>
\]

After the application of the QFT to the state $|\psi_{k}>$ , letting
$\omega=e^{-2\pi i/N}$ we arrive at the following sate:%
\[
|\phi_{k}>=\sum_{y=0}^{N-1}\left[  \frac{a_{k}}{\sqrt{N}}\sum_{z\in A}%
\omega^{zy}+\frac{b_{k}}{\sqrt{N}}\sum_{z\notin A}\omega^{zy}\right]  |y>
\]

.

At step f) we measure this state with respect to the standard basis to yield
an integer $y\in\{0,1,...,N-1\}$ from which we can determine the period $P$
using the continued fraction method.

In a later section where we summarize the main results, we provide a table
showing the probabilities of measuring $y$ for the Amplified-QFT algorithm and
compare them against the probabilites obtained by performing the QFT and QHS algorithms.

\section{The QFT Algorithm}

The QFT\ algorithm does not include the amplitude amplification steps and
consists of the following steps:

a) Initialize the state to be the uniform superposition state $|\psi>=H|0>$
where $H$ is the Hadamard transform.

b) Apply the oracle $f$ to $|\psi>$

c) Apply the QFT to this state

d)\ Make a measurement $y$

e) Form the continued fraction expansion of $y/N$ to find $d/P$

f) If $|y/N-d/P|<1/2N^{2}$ and $\gcd(d,P)=1$ then $\ P$ is recovered

g) If $\gcd(d,P)\neq1$ repeat the algorithm starting at step a)

At step b) after applying the oracle the state is given by (See later sections
for the detailed analysis of the QFT algorithm):%

\[
|\psi_{1}>=\frac{1}{\sqrt{N}}\left[  (-2)\sum_{z\in A}|z>+\sum_{z=0}%
^{N-1}|z>\right]
\]

At step c) the QFT applies the following action:%
\[
|z>\rightarrow\frac{1}{\sqrt{N}}\sum_{y=0}^{N-1}\omega^{zy}|y>
\]

\noindent to get%
\[
|\psi_{2}>=\sum_{y=0}^{N-1}\left[  \frac{(-2)}{N}\sum_{z\in A}\omega
^{zy}+\frac{1}{N}\sum_{z=0}^{N-1}\omega^{zy}\right]  |y>
\]

At step d) we measure this state with respect to the standard basis to yield
an integer $y\in\{0,1,...,N-1\}$ from which we can determine the period $P$
using the continued fraction method.

\section{The QHS Algorithm}

The QHS\ algorithm is a two register algorithm and does not include the
amplitude amplification steps. It consists of the following steps:

a) Initialize the state to be the uniform superposition state $|\psi>=H|0>|0>$
where $H$ is the Hadamard transform.

b) Apply the oracle $f$ and put the result into the second register of
$|\psi>$

c) Apply the QFT to the first register of this state

d)\ Make a measurement $y$

e) Form the continued fraction expansion of $y/N$ to find $d/P$

f) If $|y/N-d/P|<1/2N^{2}$ and $\gcd(d,P)=1$ then $\ P$ is recovered

g) If $\gcd(d,P)\neq1$ repeat the algorithm starting at step a)

At step b) we have the following state (See later sections for the detailed
analysis of the QHS algorithm):%

\[
|\psi_{1}>=\frac{1}{\sqrt{N}}%
%TCIMACRO{\dsum \limits_{x=0}^{N-1}}%
%BeginExpansion
{\displaystyle\sum\limits_{x=0}^{N-1}}
%EndExpansion
|x>|f(x)>
\]

After applying the QFT the state is given by:%

\[
|\psi_{2}>=%
%TCIMACRO{\dsum \limits_{y=0}^{N-1}}%
%BeginExpansion
{\displaystyle\sum\limits_{y=0}^{N-1}}
%EndExpansion
\frac{1}{N}|y>\sum_{x=0}^{N-1}\omega^{xy}|f(x)>
\]

At step d) we measure this state with respect to the standard basis to yield
an integer $y\in\{0,1,...,N-1\}$ from which we can determine the period $P$
using the continued fraction method.

\section{Summary of\ the Main Results}

We summarize the main results and compare the probability $\Pr(y)$ of
measuring $y$ in the final state arrived at for each of the three algorithms:
1) the Amplified-QFT algorithm 2) the QFT algorithm and 3) the QHS algorithm.
Here $\sin\theta=\sqrt{M/N}$ and $k=\left\lfloor \frac{\pi}{4\theta
}\right\rfloor $ and $0\leq$ $\frac{\sin^{2}(\pi MPy/N)}{\sin^{2}(\pi
Py/N)}\leq M^{2}$.

Case 1 (Amplified-QFT):

The probability $\Pr(y)$ is given exactly by

\begin{center}%
\[
\left\{
\begin{tabular}
[c]{lll}%
$\cos^{2}2k\theta$ & if & $y=0$\\
&  & \\
$tan^{2}\theta\sin^{2}2k\theta$ & if & $Py=0\operatorname{mod}N,y\neq0$\\
&  & \\
$\frac{1}{M^{2}}tan^{2}\theta\sin^{2}2k\theta\frac{\sin^{2}(\pi MPy/N)}%
{\sin^{2}(\pi Py/N)}$ & if & $Py\neq0\operatorname{mod}N\text{and }%
MPy\neq0\operatorname{mod}N$\\
&  & \\
$0$ & if & $Py\neq0\operatorname{mod}N\text{ and }MPy=0\operatorname{mod}N$%
\end{tabular}
\ \ \ \ \ \right\}
\]

\end{center}

Case 2 (QFT):

\bigskip The probability $\Pr(y)$ is given exactly by%

\[
\left\{
\begin{tabular}
[c]{lll}%
$\left(  1-\frac{2M}{N}\right)  ^{2}$ & if & $y=0$\\
&  & \\
$4\frac{M^{2}}{N^{2}}$ & if & $Py=0\operatorname{mod}N,y\neq0$\\
&  & \\
$\frac{4}{N^{2}}\frac{\sin^{2}(\pi MPy/N)}{\sin^{2}(\pi Py/N)}$ & if &
$Py\neq0\operatorname{mod}N\text{and }MPy\neq0\operatorname{mod}N$\\
&  & \\
$0$ & if & $Py\neq0\operatorname{mod}N\text{ and }MPy=0\operatorname{mod}N$%
\end{tabular}
\ \ \ \ \ \right\}
\]

Case 3 (QHS):

\bigskip The probability $\Pr(y)$ is given exactly by%

\[
\left\{
\begin{tabular}
[c]{lll}%
$1-\frac{2M(N-M)}{N^{2}}$ & if & $y=0$\\
&  & \\
$\frac{2M^{2}}{N^{2}}$ & if & $Py=0\operatorname{mod}N,y\neq0$\\
&  & \\
$\frac{2}{N^{2}}\frac{\sin^{2}(\pi MPy/N)}{\sin^{2}(\pi Py/N)}$ & if &
$Py\neq0\operatorname{mod}N\text{and }MPy\neq0\operatorname{mod}N$\\
&  & \\
$0$ & if & $Py\neq0\operatorname{mod}N\text{ and }MPy=0\operatorname{mod}N$%
\end{tabular}
\ \ \ \ \ \ \right\}
\]

We note that for the QFT and QHS algorithms $\Pr(y=0)$ is very close to $1$
because $M<<N.$ In the cases where $y\neq0$ we compare the ratios of $\Pr(y)$
\ in the Amplified-QFT and QFT case and then in the Amplified-QFT and
QHS\ case. Let $y$ be fixed such that either

1. $Py=0\operatorname{mod}N,y\neq0$ or

2. $Py\neq0\operatorname{mod}N$and $MPy\neq0\operatorname{mod}N$

and define $\Pr Ratio(y)=\Pr(y)_{Amplified-QFT}/\Pr(y)_{QFT}$ \ then we have
the following (see the later detailed sections)%

\begin{align*}
\frac{N}{4M}(\frac{N}{N-M}) &  \geq\Pr Ratio(y)\geq\frac{N}{4M}(\frac{N}%
{N-M})(1-\frac{2M}{N})^{2}\\
&  \Longrightarrow\Pr Ratio(y)\approx\frac{N}{4M}%
\end{align*}

and define $\Pr Ratio(y)=\Pr(y)_{Amplified-QFT}/\Pr(y)_{QHS}$ \ then we have
the following%

\begin{align*}
\frac{N}{2M}(\frac{N}{N-M})  &  \geq\Pr Ratio(y)\geq\frac{N}{2M}(\frac{N}%
{N-M})(1-\frac{2M}{N})^{2}\\
&  \Longrightarrow\Pr Ratio(y)\approx\frac{N}{2M}%
\end{align*}

Let $S_{ALG}=\{y:|\frac{y}{N}-\frac{d}{P}|\leq\frac{1}{2P^{2}},(d,P)=1\}$ be
the set of "successful" $y$'s. That is $S_{ALG\text{ }}$consists of those
$y$'s which can be measured after applying one of the three algorithms denoted
by $ALG$ and from which the period $P$ can be recovered by the method of
continued fractions. Note that the set $S_{ALG}$ is the same for each
algorithm. However the probability of this set varies with each algorithm. We
can see from the following that given $y1$ and $y2$, whose probability ratios
satisfy the same inequality, we can add their probabilities to get a new ratio
that satisfies the same inequality. In this way we can add probabilities over
a set on the numerator and denominator and maintain the inequality:
\begin{align*}
A  &  >\frac{P(y1)}{Q(y1)}>B\text{ }and\text{ }A>\frac{P(y2)}{Q(y2)}>B\\
&  \Longrightarrow A>\frac{P(y1)+P(y2)}{Q(y1)+Q(y2)}>B
\end{align*}

\noindent We see from the cases given above that%
\[
\frac{N}{4M}(\frac{N}{N-M})\geq\frac{\Pr(S_{Amplified-QFT})}{\Pr(S_{QFT})}%
\geq\frac{N}{4M}(\frac{N}{N-M})(1-\frac{2M}{N})^{2}%
\]

\noindent where the difference between the upper bound and lower bound is
exactly 1 and that%

\[
\frac{N}{2M}(\frac{N}{N-M})\geq\frac{\Pr(S_{Amplified-QFT})}{\Pr(S_{QHS})}%
\geq\frac{N}{2M}(\frac{N}{N-M})(1-\frac{2M}{N})^{2}%
\]

\noindent where the difference between the upper bound and lower bound is
exactly 2.

This shows that the Amplified-QFT is approximately$\frac{N}{4M}$ times more
successful than the QFT and $\frac{N}{2M}$ times more successful than the QHS
when $M<<N$. In addition it also shows that the QFT\ is $2$ times more
successful than the QHS in this problem. However, the success of the
Amplified-QFT algorithms comes at an increase in work factor of $O(\sqrt
{\frac{N}{M}})$. We note that in the case that P is a prime number that
$(d,P)=1$ is met trivially. However when P is composite the algorithms may
need to be rerun several times until $(d,P)=1$ is satisfied.

Towards the end of the paper we show how to test whether a putative value of
$P$, given $s$ is known, can be tested to see if it is the correct value. We
also investigate the case where $s$ is unknown but is from a small known set
of values such that the values of $s$ can be exhausted over on a classical
computer. We also show how $s$ can be recovered by using a quantum algorithm
using amplitude amplification followed by a measurement.

\section{\textbf{The Amplified-QFT\ is Quadratically Faster than the QFT\ or
the QHS }}

We show that the Amplified-QFT algorithm is, on average, quadratically faster
than the QFT\ or QHS\ algorithms. In order to show this, we use the geometric
probability distribution which provides the probability of the first success
in a sequence of trials where the probability of success is $p$ and the
probability of failure is \thinspace$1-p.$ \ For both the QFT\ and\ QHS
algorithms a trial is one complete execution of the algorithm. Because the
probability of measuring $y=0$ is close to $1$ we expect to have to repeat the
algorithm many times due to failure of measuring a successful $y$, before we
have the first success.

If $X$ is the random variable which counts the number of trials until the
first success then%

\[
P(X=k)=(1-p)^{k-1}p\text{ for }k=1,2...
\]

The expected value $E[X]$ and variance $Var[X]$ are given by:%

\[
E[X]=\frac{1}{p}\text{ and }Var[X]=\frac{1-p}{p^{2}}%
\]

The workfactor of the Amplified-QFT algorithm is given by the number of
iterations of each amplification step followed by a single QFT step:%

\[
O(\sqrt{\frac{N}{M}})
\]

For the QFT algorithm we have the probability of failure $\ 1-p$ is given by%

\[
\Pr(failure)=1-p\geq\Pr(y=0)=(1-\frac{2M}{N})^{2}%
\]

then%

\[
\Pr(success)=p\leq1-(1-\frac{2M}{N})^{2}=\frac{4M}{N}(1-\frac{M}{N})
\]

Then for the QFT algorithm, the expected number of trials until the first
success is%

\[
E[X]=\frac{1}{p}\geq\frac{N}{4M(1-\frac{M}{N})}\geq\frac{N}{4M}%
\]

The workfactor of the QFT\ algorithm is the expected number of times the QFT
has to be run, is given approximately by:%

\[
O(\frac{N}{M})
\]

Therefore the ratio of the expected work factor of the QFT algorithm and the
work factor of the Amplified-QFT is given by%

\[
O(\sqrt{\frac{N}{M}})
\]

showing that the Amplified-QFT\ algorithm is, on average, quadratically faster
than the QFT algorithm.

The variance in the number of times the QFT algorithm is run is given by%

\[
Var[X]=\frac{1-p}{p^{2}}\geq(\frac{N}{N-M})^{2}(\frac{N-2M}{4M})^{2}%
\]

For the QHS algorithm we have the probability of failure $\ 1-p$ is given by%

\[
\Pr(failure)=1-p\geq\Pr(y=0)=1-\frac{2M(N-M)}{N^{2}}%
\]

then%

\[
\Pr(success)=p\leq1-(1-\frac{2M(N-M)}{N^{2}})=\frac{2M}{N}(1-\frac{M}{N})
\]

Then for the QHS algorithm, the expected number of trials until the first
success is%

\[
E[X]=\frac{1}{p}\geq\frac{N}{2M(1-\frac{M}{N})}\geq\frac{N}{2M}%
\]

The workfactor of the QHS\ algorithm is the expected number of times the QHS
has to be run, is given approximately by:%

\[
O(\frac{N}{M})
\]

Therefore the ratio of the expected work factor of the QHS algorithm and the
work factor of the Amplified-QFT is given by%

\[
O(\sqrt{\frac{N}{M}})
\]

showing that the Amplified-QFT\ algorithm is, on average, quadratically faster
than the QHS algorithm.

The variance in the number of times the QHS algorithm is run is given by%

\[
Var[X]=\frac{1-p}{p^{2}}\geq(\frac{N}{N-M})^{2}(\frac{(N-M)^{2}+M^{2}}{4M^{2}%
})
\]

\section{\textbf{The Amplified-QFT Algorithm - Detailed Analysis}}

In this section we examine the Amplified-QFT algorithm in detail and produce
the results for the probability of success that were summarized earlier in the paper.

The Amplified-QFT algorithm is defined by the following procedure (see earlier section):

\emph{Steps a) to d):} Apply the Amplitude Amplification algorithm to the
starting state $|0>$. The resulting state is given by $|\psi_{k}>$ (ref[4],
ref[7],ref[1]) where $k=\left\lfloor \frac{\pi}{4\sin^{-1}(\sqrt{M/N}%
)}\right\rfloor $:%
\[
|\psi_{k}>=a_{k}\sum_{z\in A}|z>+b_{k}\sum_{z\notin A}|z>
\]

\noindent where%

\[
a_{k}=\frac{1}{\sqrt{M}}\sin(2k+1)\theta,b_{k}=\frac{1}{\sqrt{N-M}}%
\cos(2k+1)\theta
\]

\noindent are the appropriate amplitudes of the states and where
\[
\sin\theta=\sqrt{M/N},\cos\theta=\sqrt{1-M/N}%
\]

\noindent Now we have , ref[7],

$k=\left\lfloor \frac{\pi}{4\theta}\right\rfloor $ $\Longrightarrow$
$\frac{\pi}{4\theta}-1\leq k\leq\frac{\pi}{4\theta}$ $\Longrightarrow$
$\frac{\pi}{2}-\theta\leq(2k+1)\theta\leq\frac{\pi}{2}+\theta$

$\Longrightarrow\sin\theta=\cos(\frac{\pi}{2}-\theta)\geq\cos(2k+1)\theta
\geq\cos(\frac{\pi}{2}+\theta)=-\sin\theta$

\noindent Notice that the total probability of the N-M labels that are not in
A is%

\begin{align*}
(N-M)(\frac{1}{\sqrt{N-M}}\cos(2k+1)\theta)^{2}  &  =\cos^{2}(2k+1)\theta\\
&  \Longrightarrow\cos^{2}(2k+1)\theta\leq\sin^{2}\theta=\sin^{2}(\sin
^{-1}(\sqrt{\frac{M}{N}}))\\
&  \Longrightarrow\cos^{2}(2k+1)\theta\leq\frac{M}{N}%
\end{align*}
whereas the total probability of the M labels in A is
\begin{align*}
M(\frac{1}{\sqrt{M}}\sin(2k+1)\theta)^{2}  &  =\sin^{2}(2k+1)\theta=1-\cos
^{2}(2k+1)\theta\\
&  \Longrightarrow\sin^{2}(2k+1)\theta\geq1-\frac{M}{N}%
\end{align*}
.

\emph{Step e):} Apply the QFT which performs the following action%

\[
|z>\rightarrow\frac{1}{\sqrt{N}}\sum_{y=0}^{N-1}e^{-2\pi izy/N}|y>
\]

After the application of the QFT to the state $|\psi_{k}>$ , letting
$\omega=e^{-2\pi i/N}$ , we have%
\[
|\phi_{k}>=\frac{a_{k}}{\sqrt{N}}\sum_{z\in A}\sum_{y=0}^{N-1}\omega
^{zy}|y>+\frac{b_{k}}{\sqrt{N}}\sum_{z\notin A}\sum_{y=0}^{N-1}\omega^{zy}|y>
\]

After interchanging the order of summation, we have%
\[
|\phi_{k}>=\sum_{y=0}^{N-1}\left[  \frac{a_{k}}{\sqrt{N}}\sum_{z\in A}%
\omega^{zy}+\frac{b_{k}}{\sqrt{N}}\sum_{z\notin A}\omega^{zy}\right]  |y>
\]

.

\emph{Steps f) to i):} Measure with respect to the standard basis to yield a
integer $y\in\{0,1,...,N-1\}$ from which we can determine the period P using
the continued fraction method.

The amplitude $Amp(y)$ of $|y>$ is given by
\begin{align*}
Amp(y) &  =\frac{a_{k}}{\sqrt{N}}\sum_{z\in A}\omega^{zy}+\frac{b_{k}}%
{\sqrt{N}}\sum_{z\notin A}\omega^{zy}\\
&  =\frac{(a_{k}-b_{k})}{\sqrt{N}}\sum_{z\in A}\omega^{zy}+\frac{b_{k}}%
{\sqrt{N}}\sum_{z=0}^{N-1}\omega^{zy}\\
&  =\frac{(a_{k}-b_{k})}{\sqrt{N}}\sum_{r=0}^{M-1}\omega^{(s+rP)y}+\frac
{b_{k}}{\sqrt{N}}\sum_{z=0}^{N-1}\omega^{zy}\\
&  =\frac{(a_{k}-b_{k})}{\sqrt{N}}\omega^{sy}\sum_{r=0}^{M-1}\omega
^{rPy}+\frac{b_{k}}{\sqrt{N}}\sum_{z=0}^{N-1}\omega^{zy}%
\end{align*}

We calculate the $\Pr(y)$ for the following cases:

\qquad a) $y=0$

\qquad b) $Py=0\operatorname{mod}N$ and $y\neq0$

\qquad c) $Py\neq0\operatorname{mod}N$

\subsection{Amplified-QFT Analysis: y=0}

\noindent We have%

\begin{align*}
Amp(y)  &  =\frac{a_{k}}{\sqrt{N}}\sum_{z\in A}\omega^{zy}+\frac{b_{k}}%
{\sqrt{N}}\sum_{z\notin A}\omega^{zy}\\
&  =\frac{1}{\sqrt{N}}(Ma_{k}+(N-M)b_{k})\\
&  =\frac{1}{\sqrt{N}}\left[  \frac{M}{\sqrt{M}}\sin(2k+1)\theta+\frac
{N-M}{\sqrt{N-M}}\cos(2k+1)\theta\right] \\
&  =\sqrt{\frac{M}{N}}\sin(2k+1)\theta+\sqrt{1-\frac{M}{N}}\cos(2k+1)\theta\\
&  =\sin\theta\sin(2k+1)\theta+\cos\theta\cos(2k+1)\theta\\
&  =\cos(2k\theta)
\end{align*}

\noindent We have
\[
\Pr(y=0)=\cos^{2}(2k\theta)
\]

\subsection{Amplified-QFT Analysis: $Py=0\operatorname{mod}N,y\neq0$}

\noindent Using the fact that%
\[
\sum_{z=0}^{N-1}\omega^{zy}=\frac{1-\omega^{Ny}}{1-\omega^{y}}=0,w^{y}\neq1
\]

\noindent we have
\begin{align*}
Amp(y)  &  =\frac{(a_{k}-b_{k})}{\sqrt{N}}\omega^{sy}\sum_{r=0}^{M-1}%
\omega^{rPy}+\frac{b_{k}}{\sqrt{N}}\sum_{z=0}^{N-1}\omega^{zy}\\
&  =\frac{(a_{k}-b_{k})}{\sqrt{N}}\omega^{sy}\sum_{r=0}^{M-1}\omega^{rPy}\\
&  =\frac{(a_{k}-b_{k})}{\sqrt{N}}\omega^{sy}M\\
&  =\frac{Mw^{sy}}{\sqrt{NM}}\sin(2k+1)\theta-\frac{Mw^{sy}}{\sqrt{N(N-M)}%
}\cos(2k+1)\theta\\
&  =\omega^{sy}\sqrt{\frac{M}{N}}(\sin(2k+1)\theta-\sqrt{\frac{M/N}{1-M/N}%
}\cos(2k+1)\theta)\\
&  =\omega^{sy}\sqrt{\frac{M}{N}}(\sin(2k+1)\theta-\frac{\sin\theta}%
{\cos\theta}\cos(2k+1)\theta)\\
&  =\omega^{sy}\tan\theta\sin2k\theta
\end{align*}

\noindent We have
\[
\Pr(y)=tan^{2}\theta\sin^{2}2k\theta
\]

\noindent Using $k=\left\lfloor \frac{\pi}{4\theta}\right\rfloor $
$\Longrightarrow$ $\frac{\pi}{4\theta}-1\leq k\leq\frac{\pi}{4\theta}$
$\Longrightarrow$ $\frac{\pi}{2}-2\theta\leq2k\theta\leq\frac{\pi}%
{2}\Longrightarrow\sin(\frac{\pi}{2}-2\theta)\leq\sin2k\theta\leq1$ we have
\begin{align*}
\frac{\sin^{2}\theta}{\cos^{2}\theta}  &  \geq\Pr(y)=\tan^{2}\theta\sin
^{2}2k\theta\geq\tan^{2}\theta\sin^{2}(\frac{\pi}{2}-2\theta)\\
&  \Longrightarrow\frac{M}{N}\frac{1}{1-\frac{M}{N}}\geq\Pr(y)\geq\tan
^{2}\theta\sin^{2}(\frac{\pi}{2}-2\theta)\\
&  \Longrightarrow\frac{M}{N}(\frac{N}{N-M})\geq\Pr(y)\geq\frac{\sin^{2}%
\theta}{\cos^{2}\theta}\cos^{2}2\theta\\
&  \Longrightarrow\frac{M}{N}(\frac{N}{N-M})\geq\Pr(y)\geq\frac{\sin^{2}%
\theta}{\cos^{2}\theta}(2\cos^{2}\theta-1)^{2}\\
&  \Longrightarrow\frac{M}{N}(\frac{N}{N-M})\geq\Pr(y)\geq\frac{M}{N}(\frac
{N}{N-M})(1-\frac{2M}{N})^{2}%
\end{align*}

\subsection{Amplified-QFT Analysis: $Py\neq0\operatorname{mod}N$}

\noindent Making use of the previous results we have
\begin{align*}
Amp(y)  &  =\frac{(a_{k}-b_{k})}{\sqrt{N}}\omega^{sy}\sum_{r=0}^{M-1}%
\omega^{rPy}+\frac{b_{k}}{\sqrt{N}}\sum_{z=0}^{N-1}\omega^{zy}\\
&  =\frac{(a_{k}-b_{k})}{\sqrt{N}}\omega^{sy}\sum_{r=0}^{M-1}\omega^{rPy}\\
&  =\frac{(a_{k}-b_{k})}{\sqrt{N}}\omega^{sy}\left[  \frac{1-\omega^{MPy}%
}{1-\omega^{Py}}\right] \\
&  =\frac{1}{M}\frac{(a_{k}-b_{k})}{\sqrt{N}}\omega^{sy}M\left[
\frac{1-\omega^{MPy}}{1-\omega^{Py}}\right] \\
&  =\frac{1}{M}\omega^{sy}\tan\theta\sin2k\theta\left[  \frac{1-\omega^{MPy}%
}{1-\omega^{Py}}\right]
\end{align*}

\noindent\noindent\noindent Making use of the following identity%
\[
|1-e^{i\theta}|^{2}=4\sin^{2}(\theta/2)
\]

\noindent we have%
\[
\left\vert \frac{1-\omega^{MPy}}{1-\omega^{Py}}\right\vert ^{2}=\frac{\sin
^{2}(\pi MPy/N)}{\sin^{2}(\pi Py/N)}%
\]

\noindent and so
\[
\Pr(y)=\frac{1}{M^{2}}tan^{2}\theta\sin^{2}2k\theta\frac{\sin^{2}(\pi
MPy/N)}{\sin^{2}(\pi Py/N)}%
\]

\noindent Using the previous result $\frac{M}{N}(\frac{N}{N-M})\geq\tan
^{2}\theta\sin^{2}2k\theta\geq\frac{M}{N}(\frac{N}{N-M})(\frac{N-2M}{N})^{2}$
and letting $R=\frac{\sin^{2}(\pi MPy/N)}{\sin^{2}(\pi Py/N)}$ we have

\bigskip%
\begin{align*}
\frac{1}{M^{2}}\frac{M}{N}(\frac{N}{N-M})R  &  \geq\Pr(y)\geq\frac{1}{M^{2}%
}\frac{M}{N}(\frac{N}{N-M})(1-\frac{2M}{N})^{2}R\text{ and so}\\
\frac{1}{NM}(\frac{N}{N-M})R  &  \geq\Pr(y)\geq\frac{1}{NM}(\frac{N}%
{N-M})(1-\frac{2M}{N})^{2}R\text{ }%
\end{align*}

\noindent We notice that if in addition $MPy=0\operatorname{mod}N$ then
$\Pr(y)=0.$

\subsection{Amplified-QFT Summary}

The probability $\Pr(y)$ is given exactly by

\begin{center}%
\[
\left\{
\begin{tabular}
[c]{lll}%
$\cos^{2}2k\theta$ & if & $y=0$\\
&  & \\
$tan^{2}\theta\sin^{2}2k\theta$ & if & $Py=0\operatorname{mod}N,y\neq0$\\
&  & \\
$\frac{1}{M^{2}}tan^{2}\theta\sin^{2}2k\theta\frac{\sin^{2}(\pi MPy/N)}%
{\sin^{2}(\pi Py/N)}$ & if & $Py\neq0\operatorname{mod}N\text{and }%
MPy\neq0\operatorname{mod}N$\\
&  & \\
$0$ & if & $Py\neq0\operatorname{mod}N\text{ and }MPy=0\operatorname{mod}N$%
\end{tabular}
\ \ \ \ \ \right\}
\]

\end{center}

\section{\textbf{The QFT Algorithm - Detailed Analysis.}}

In this section we examine the QFT algorithm in detail and produce the results
for the probability of success that were summarized earlier in the paper. We
just apply the QFT to the binary oracle f, which is 1 on A and 0 elsewhere.

We begin with the following state%

\[
|\xi>=\frac{1}{\sqrt{N}}\sum_{z=0}^{N-1}|z>\otimes\frac{1}{\sqrt{2}}(|0>-|1>)
\]
\bigskip

\noindent and apply the unitary transform for f, $U_{f}$ ,\ to this state
which performs the following action:%

\[
U_{f}|z>|c>=|z>|c\oplus f(z)>
\]

\noindent to get the state $|\psi>$%

\begin{align*}
|\psi &  >=U_{f}\frac{1}{\sqrt{N}}\sum_{z=0}^{N-1}|z>\frac{1}{\sqrt{2}%
}(|0>-|1>)\\
&  =\frac{1}{\sqrt{N}}\left[  (-1)\sum_{z\in A}|z>+\sum_{z\notin A}|z>\right]
\frac{1}{\sqrt{2}}(|0>-|1>)\\
&  =\frac{1}{\sqrt{N}}\left[  (-2)\sum_{z\in A}|z>+\sum_{z=0}^{N-1}|z>\right]
\frac{1}{\sqrt{2}}(|0>-|1>)
\end{align*}

\noindent Next we apply the QFT to try to find the period P, dropping
$\frac{1}{\sqrt{2}}(|0>-|1>)$.

\noindent The QFT applies the following action:%
\[
|z>\rightarrow\frac{1}{\sqrt{N}}\sum_{y=0}^{N-1}\omega^{zy}|y>
\]

\noindent to get%
\[
|\phi>=\sum_{y=0}^{N-1}\left[  \frac{(-2)}{N}\sum_{z\in A}\omega^{zy}+\frac
{1}{N}\sum_{z=0}^{N-1}\omega^{zy}\right]  |y>
\]

We calculate the $\Pr(y)$ for the following cases:

\qquad a) $y=0$

\qquad b) $Py=0\operatorname{mod}N$ and $y\neq0$

\qquad c) $Py\neq0\operatorname{mod}N$

\subsection{QFT Analysis: $y=0$}

We have%

\begin{align*}
Amp(y)  &  =\frac{(-2)}{N}\sum_{z\in A}\omega^{zy}+\frac{1}{N}\sum_{z=0}%
^{N-1}\omega^{zy}\\
&  =\frac{(-2)M}{N}+\frac{N}{N}\\
&  =1-\frac{2M}{N}%
\end{align*}

\noindent Therefore, in the QFT case, we have $\Pr(y=0)$ is very close to 1
and is given by
\[
\Pr(y=0)=1-\frac{4M}{N}+4\frac{M^{2}}{N^{2}}=\left(  1-\frac{2M}{N}\right)
^{2}%
\]

\noindent whereas in the Amplified-QFT case we have $\Pr(y=0)$ is given by
\[
\Pr(y=0)=\cos^{2}2k\theta
\]

\subsection{QFT Analysis: $Py=0\operatorname{mod}N,y\neq0$}

Using the fact that%
\[
\sum_{z=0}^{N-1}\omega^{zy}=\frac{1-\omega^{Ny}}{1-\omega^{y}}=0
\]

\noindent we have
\begin{align*}
Amp(y)  &  =\ \frac{-2}{N}\sum_{z\in A}\omega^{zy}+\frac{1}{N}\sum_{z=0}%
^{N-1}\omega^{zy}\\
&  =\frac{-2}{N}\omega^{sy}\sum_{r=0}^{M-1}\omega^{rPy}\\
&  =\frac{-2M}{N}\omega^{sy}%
\end{align*}

\noindent Therefore in the QFT$\ $case we have $\Pr(y)$ is given by
\[
\Pr(y)=4\frac{M^{2}}{N^{2}}%
\]

\noindent whereas in the Amplified-QFT case we have $\Pr(y)$ is given by \
\[
\Pr(y)=tan^{2}\theta\sin^{2}2k\theta
\]

\noindent We can determine how the increase in amplitude varies with the
number of iterations $k$ of the Grover step in the Amplified-QFT by examining
the ratio of the amplitudes of the Amplified-QFT case and QFT\ case. This
ratio is given exactly by
\begin{align*}
AmpRatio(y)  &  =\frac{\frac{(a_{k}-b_{k})}{\sqrt{N}}\omega^{sy}M}{\frac
{-2M}{N}\omega^{sy}}\\
&  =\frac{(a_{k}-b_{k})}{-2}\sqrt{N}\\
&  =\frac{1}{-2}\left[  \sqrt{\frac{N}{M}}\sin(2k+1)\theta-\sqrt{\frac{N}%
{N-M}}\cos(2k+1)\theta\right] \\
&  =\frac{N}{-2M}\tan\theta\sin2k\theta
\end{align*}

\noindent Using $k=\left\lfloor \frac{\pi}{4\theta}\right\rfloor $ and making
use of $\frac{M}{N}(\frac{N}{N-M})\geq\tan^{2}\theta\sin^{2}2k\theta\geq
\frac{M}{N}(\frac{N}{N-M})(\frac{N-2M}{N})^{2}$, we have the following
inequality for the $\Pr Ratio(y)$, the increase in the probability due to amplification:%

\begin{align*}
\frac{N}{4M}(\frac{N}{N-M})  &  \geq\Pr Ratio(y)\geq\frac{N}{4M}(\frac{N}%
{N-M})(1-\frac{2M}{N})^{2}\\
&  \Longrightarrow\Pr Ratio(y)\approx\frac{N}{4M}%
\end{align*}

\subsection{QFT Analysis: $Py\neq0\operatorname{mod}N$}

We have
\begin{align*}
Amp(y)  &  =\ \frac{-2}{N}\sum_{z\in A}\omega^{zy}+\frac{1}{N}\sum_{z=0}%
^{N-1}\omega^{zy}\\
&  =\frac{-2}{N}w^{sy}\sum_{r=0}^{M-1}\omega^{rPy}\\
&  =\frac{-2}{N}w^{sy}\left[  \frac{1-\omega^{MPy}}{1-\omega^{Py}}\right] \\
&  =\frac{-2}{N}w^{sy}\left[  \frac{1-\omega^{MPy}}{1-\omega^{Py}}\right]
\end{align*}

\noindent\noindent Once again, making use of the following identity%
\[
|1-e^{i\theta}|^{2}=4\sin^{2}(\theta/2)
\]

\noindent in the QFT\ case, we have $\Pr(y)$ is given by%

\[
\Pr(y)=\frac{4}{N^{2}}\left[  \frac{\sin^{2}(\pi MPy/N)}{\sin^{2}(\pi
Py/N)}\right]
\]

\noindent whereas in the Amplified-QFT case we have $\Pr(y)$ is given by
\[
\Pr(y)=\frac{1}{M^{2}}tan^{2}\theta\sin^{2}2k\theta\frac{\sin^{2}(\pi
MPy/N)}{\sin^{2}(\pi Py/N)}%
\]

\noindent We notice that if in addition $MPy=0\operatorname{mod}N$ then
$\Pr(y)=0.$

\noindent The ratio of the amplitudes of the Amplified-QFT case and QFT\ case
is given exactly by
\begin{align*}
AmpRatio(y)  &  =\frac{\frac{(a_{k}-b_{k})}{\sqrt{N}}\omega^{sy}\left[
\frac{1-\omega^{MPy}}{1-\omega^{Py}}\right]  }{\frac{-2}{N}w^{sy}\left[
\frac{1-\omega^{MPy}}{1-\omega^{Py}}\right]  }\\
&  =\frac{(a_{k}-b_{k})}{-2}\sqrt{N}\\
&  =\frac{1}{-2}\left[  \sqrt{\frac{N}{M}}\sin(2k+1)\theta-\sqrt{\frac{N}%
{N-M}}\cos(2k+1)\theta\right] \\
&  =\frac{N}{-2M}\tan\theta\sin2k\theta
\end{align*}

\noindent We note that this ratio is the same as in that given in the previous
section and is independent of $y$. The variables in this ratio do not depend
in anyway on the QFT.

As in the previous section, we have the following inequality for the $\Pr
Ratio(y)$, the increase in the probability due to amplification when
$k=\left\lfloor \frac{\pi}{4\theta}\right\rfloor $ and making use of $\frac
{M}{N}(\frac{N}{N-M})\geq\tan^{2}\theta\sin^{2}2k\theta\geq\frac{M}{N}%
(\frac{N}{N-M})(\frac{N-2M}{N})^{2}$%

\begin{align*}
\frac{N}{4M}(\frac{N}{N-M})  &  \geq\Pr Ratio(y)\geq\frac{N}{4M}(\frac{N}%
{N-M})(1-\frac{2M}{N})^{2}\\
&  \Longrightarrow\Pr Ratio(y)\approx\frac{N}{4M}%
\end{align*}

\subsection{QFT Summary}

The probability $\Pr(y)$ is given exactly by%

\[
\left\{
\begin{tabular}
[c]{lll}%
$\left(  1-\frac{2M}{N}\right)  ^{2}$ & if & $y=0$\\
&  & \\
$4\frac{M^{2}}{N^{2}}$ & if & $Py=0\operatorname{mod}N,y\neq0$\\
&  & \\
$\frac{4}{N^{2}}\frac{\sin^{2}(\pi MPy/N)}{\sin^{2}(\pi Py/N)}$ & if &
$Py\neq0\operatorname{mod}N\text{and }MPy\neq0\operatorname{mod}N$\\
&  & \\
$0$ & if & $Py\neq0\operatorname{mod}N\text{ and }MPy=0\operatorname{mod}N$%
\end{tabular}
\ \ \ \ \right\}
\]

\section{\textbf{The QHS Algorithm - Detailed Analysis}}

In this section we examine the QHS algorithm in detail and produce the results
for the probability of success that were summarized earlier in the paper. The
QHS algorithm is a two register algorithm as follows (see ref[13] for
details). We begin with $|0>|0>$ where the first register is $n$ qubits and
the second register is $1$ qubit and apply the Hadamard transform to the first
register to get a uniform superposition state, followed by the unitary
transformation for the Oracle f to get:%

\[
|\psi>=\frac{1}{\sqrt{N}}%
%TCIMACRO{\dsum \limits_{x=0}^{N-1}}%
%BeginExpansion
{\displaystyle\sum\limits_{x=0}^{N-1}}
%EndExpansion
|x>|f(x)>
\]
Next we apply the QFT to the first register to get%

\begin{align*}
|\psi &  >=\frac{1}{\sqrt{N}}%
%TCIMACRO{\dsum \limits_{x=0}^{N-1}}%
%BeginExpansion
{\displaystyle\sum\limits_{x=0}^{N-1}}
%EndExpansion
\frac{1}{\sqrt{N}}\sum_{y=0}^{N-1}\omega^{xy}|y>|f(x)>\\
&  =%
%TCIMACRO{\dsum \limits_{y=0}^{N-1}}%
%BeginExpansion
{\displaystyle\sum\limits_{y=0}^{N-1}}
%EndExpansion
\frac{1}{N}\sum_{x=0}^{N-1}\omega^{xy}|y>|f(x)>\\
&  =%
%TCIMACRO{\dsum \limits_{y=0}^{N-1}}%
%BeginExpansion
{\displaystyle\sum\limits_{y=0}^{N-1}}
%EndExpansion
\frac{1}{N}|y>\sum_{x=0}^{N-1}\omega^{xy}|f(x)>\\
&  =\
%TCIMACRO{\dsum \limits_{y=0}^{N-1}}%
%BeginExpansion
{\displaystyle\sum\limits_{y=0}^{N-1}}
%EndExpansion
\frac{|||\Gamma(y)>||}{N}|y>\frac{|\Gamma(y)>}{|||\Gamma(y)>||}%
\end{align*}

\noindent where
\begin{align*}
|\Gamma(y)  &  >=\sum_{x=0}^{N-1}\omega^{xy}|f(x)>\\
&  =\sum_{x\in A}^{{}}\omega^{xy}|1>+\sum_{x\notin A}^{{}}\omega^{xy}|0>
\end{align*}

\noindent and where
\[
|||\Gamma(y)>||^{2}=\left\vert \sum_{x\in A}^{{}}\omega^{xy}\right\vert
^{2}+\left\vert \sum_{x\notin A}^{{}}\omega^{xy}\right\vert ^{2}%
\]

\noindent Next we make a measurement to get $y$ and find that the probability
of this measurement is%
\begin{align*}
\Pr(y) &  =\frac{|||\Gamma(y)>||^{2}}{N^{2}}\\
&  =\frac{1}{N^{2}}\left\vert \sum_{x\in A}^{{}}\omega^{xy}\right\vert
^{2}+\frac{1}{N^{2}}\left\vert \sum_{x\notin A}^{{}}\omega^{xy}\right\vert
^{2}%
\end{align*}

\noindent The state that we end up in is of the form
\[
|\phi>=|y>\frac{|\Gamma(y)>}{|||\Gamma(y)>||}%
\]

We calculate the $\Pr(y)$ for the following cases:

\qquad a) $y=0$

\qquad b) $Py=0\operatorname{mod}N$ and $y\neq0$

\qquad c) $Py\neq0\operatorname{mod}N$

\subsection{QHS Analysis: $y=0$}

We have
\begin{align*}
\Pr(y)  &  =\frac{1}{N^{2}}\left\vert \sum_{x\in A}^{{}}\omega^{xy}\right\vert
^{2}+\frac{1}{N^{2}}\left\vert \sum_{x\notin A}^{{}}\omega^{xy}\right\vert
^{2}\\
&  =\frac{M^{2}}{N^{2}}+\frac{(N-M)^{2}}{N^{2}}=\frac{M^{2}+N^{2}-2NM+M^{2}%
}{N^{2}}\\
&  =1-\frac{2M(N-M)}{N^{2}}%
\end{align*}

\noindent whereas in the Amplified-QFT case we have $\Pr(y=0)$ is given by
\[
\Pr(y=0)=\cos^{2}2k\theta
\]

\subsection{QHS Analysis: $Py=0\operatorname{mod}N,y\neq0$}

We have%
\begin{align*}
\Pr(y)  &  =\frac{1}{N^{2}}\left\vert \sum_{x\in A}^{{}}\omega^{xy}\right\vert
^{2}+\frac{1}{N^{2}}\left\vert \sum_{x\notin A}^{{}}\omega^{xy}\right\vert
^{2}\\
&  =\frac{1}{N^{2}}\left\vert \omega^{sy}\sum_{r=0}^{M-1}\omega^{rPy}%
\right\vert ^{2}+\frac{1}{N^{2}}\left\vert \sum_{x\notin A}^{{}}\omega
^{xy}\right\vert ^{2}\\
&  =\frac{1}{N^{2}}\left\vert \omega^{sy}\sum_{r=0}^{M-1}\omega^{rPy}%
\right\vert ^{2}+\frac{1}{N^{2}}\left\vert -\omega^{sy}\sum_{r=0}^{M-1}%
\omega^{rPy}+\frac{1}{N}\sum_{x=0}^{N-1}\omega^{xy}\right\vert ^{2}\\
&  =\frac{2M^{2}}{N^{2}}%
\end{align*}

\noindent\noindent where we have used the fact that
\[
\sum_{x=0}^{N-1}\omega^{xy}=0
\]

In the Amplified-QFT case we have $\Pr(y)$ is given by \
\[
\Pr(y)=tan^{2}\theta\sin^{2}2k\theta
\]

By comparing the results of the QHS and the Amplified-QFT algorithms we have
the following inequality for the $\Pr Ratio(y)=\Pr(y)_{Amplified-QFT}%
/\Pr(y)_{QHS}$, the increase in the probability due to amplification when
$k=\left\lfloor \frac{\pi}{4\theta}\right\rfloor $ and making use of $\frac
{M}{N}(\frac{N}{N-M})\geq\tan^{2}\theta\sin^{2}2k\theta\geq\frac{M}{N}%
(\frac{N}{N-M})(\frac{N-2M}{N})^{2}$%

\begin{align*}
\frac{N}{2M}(\frac{N}{N-M})  &  \geq\Pr Ratio(y)\geq\frac{N}{2M}(\frac{N}%
{N-M})(1-\frac{2M}{N})^{2}\\
&  \Longrightarrow\Pr Ratio(y)\approx\frac{N}{2M}%
\end{align*}

\subsection{QHS Analysis: $Py\neq0\operatorname{mod}N$}

We have%
\begin{align*}
\Pr(y)  &  =\frac{1}{N^{2}}\left\vert \sum_{x\in A}^{{}}\omega^{xy}\right\vert
^{2}+\frac{1}{N^{2}}\left\vert \sum_{x\notin A}^{{}}\omega^{xy}\right\vert
^{2}\\
&  =\frac{1}{N^{2}}\left\vert \omega^{sy}\sum_{r=0}^{M-1}\omega^{rPy}%
\right\vert ^{2}+\frac{1}{N^{2}}\left\vert \sum_{x\notin A}^{{}}\omega
^{xy}\right\vert ^{2}\\
&  =\frac{1}{N^{2}}\left\vert \omega^{sy}\sum_{r=0}^{M-1}\omega^{rPy}%
\right\vert ^{2}+\frac{1}{N^{2}}\left\vert -\omega^{sy}\sum_{r=0}^{M-1}%
\omega^{rPy}+\frac{1}{N}\sum_{x=0}^{N-1}\omega^{xy}\right\vert ^{2}\\
&  =\frac{1}{N^{2}}\left\vert \omega^{sy}\left[  \frac{1-\omega^{MPy}%
}{1-\omega^{Py}}\right]  \right\vert ^{2}+\frac{1}{N^{2}}\left\vert
-\omega^{sy}\left[  \frac{1-\omega^{MPy}}{1-\omega^{Py}}\right]  \right\vert
^{2}\\
&  =\frac{2}{N^{2}}\frac{\sin^{2}(\pi MPy/N)}{\sin^{2}(\pi Py/N)}%
\end{align*}

\noindent\noindent where we have used the fact that
\[
\sum_{x=0}^{N-1}\omega^{xy}=0
\]

\noindent and that%

\[
|1-e^{i\theta}|^{2}=4\sin^{2}(\theta/2)
\]
In the Amplified-QFT case we have $\Pr(y)$ is given by
\[
\Pr(y)=\frac{1}{M^{2}}tan^{2}\theta\sin^{2}2k\theta\frac{\sin^{2}(\pi
MPy/N)}{\sin^{2}(\pi Py/N)}%
\]
We notice that if in addition $MPy=0\operatorname{mod}N$ then $\Pr(y)=0.$

By comparing the results of the QHS and the Amplified-QFT algorithms we have
the following inequality for the $\Pr Ratio(y)=\Pr(y)_{Amplified-QFT}%
/\Pr(y)_{QHS}$, the increase in the probability due to amplification when
$k=\left\lfloor \frac{\pi}{4\theta}\right\rfloor $ and making use of $\frac
{M}{N}(\frac{N}{N-M})\geq\tan^{2}\theta\sin^{2}2k\theta\geq\frac{M}{N}%
(\frac{N}{N-M})(\frac{N-2M}{N})^{2}$%

\begin{align*}
\frac{N}{2M}(\frac{N}{N-M})  &  \geq\Pr Ratio(y)\geq\frac{N}{2M}(\frac{N}%
{N-M})(1-\frac{2M}{N})^{2}\\
&  \Longrightarrow\Pr Ratio(y)\approx\frac{N}{2M}%
\end{align*}

\subsection{QHS Summary}

The $\Pr(y)$ in the QHS case is:%

\[
\left\{
\begin{tabular}
[c]{lll}%
$1-\frac{2M(N-M)}{N^{2}}$ & if & $y=0$\\
&  & \\
$\frac{2M^{2}}{N^{2}}$ & if & $Py=0\operatorname{mod}N,y\neq0$\\
&  & \\
$\frac{2}{N^{2}}\frac{\sin^{2}(\pi MPy/N)}{\sin^{2}(\pi Py/N)}$ & if &
$Py\neq0\operatorname{mod}N\text{and }MPy\neq0\operatorname{mod}N$\\
&  & \\
$0$ & if & $Py\neq0\operatorname{mod}N\text{ and }MPy=0\operatorname{mod}N$%
\end{tabular}
\ \ \ \ \ \ \ \right\}
\]

\section{\textbf{Recovering the Period P and the Offset s}}

As in Shor's algorithm, we use the continued fraction expansion of $y/N$ to
find the period $P,$where $y$ is a measured value such that $y/N$ is close to
$d/P$ and $(d,P)=1$ . See ref[2] and ref[3]for details which we provide below.

Let$\{a\}_{N}$ be the residue of $a\operatorname{mod}N$ of smallest magnitude
such that $-N/2<\{a\}_{N}<N/2.$ Let $S_{N}=\{0,1,...,N-1\}$, $S_{P}=\{d\in
S_{N}:0\leq d<P\}$ and $Y=\{y\in S_{N}:|Py|\leq P/2\}$. Then the map
$Y\rightarrow S_{P}$ given by $y\rightarrow d=d(y)=round(Py/N)$ with inverse
$y=y(d)=round(Nd/P)$ is a bijection and $\{Py\}_{N}=Py-Nd(y)$. In addition the
following two sets are in 1-1 correspondence $\{y/N:y\in Y\}$ and $\{d/P:0\leq
d<P\}.$

We make use of the following theorem from the theory of continued fractions
ref[5] (Theorem 184 p.153):

\begin{theorem}
Let $x$ be a real number and let $a$ and $b$ be integers with $b>0$. If
$|x-\frac{a}{b}|\leq\frac{1}{2b^{2}}$ then the rational $a/b$ is a convergent
of the continued fraction expansion of $x$.
\end{theorem}

\begin{corollary}
If $P^{2}\leq N$ and $|\{Py\}_{N}|\leq\frac{P}{2}$ then $d(y)/P$ is a
convergent of the continued fraction expansion of $y/N$.
\end{corollary}

\begin{proof}
Since $\{Py\}_{N}=Py-Nd(y)$ we have

$|Py-Nd(y)|\leq\frac{P}{2}$ or

$|\frac{y}{N}-\frac{d(y)}{P}|\leq\frac{1}{2N}\leq\frac{1}{2P^{2}}$

and we can apply Theorem 1 so that $d/P$ is a convergent of the continued
fraction expansion of $y/N$.
\end{proof}

Since we know $y$ and $N$ we can find the continued fraction expansion of
$y/N$. However we also need that $(d,P)=1$ in order that $d/P$ is a convergent
and enabling us to read off $P$ directly. The probability that $(d,P)=1$ is
$\varphi(P)/P$ where $\varphi(P)$ is Euler's totient function. If $P$ is prime
we get $(d,P)=1$ trivially.

By making use of the following Theorem it can be shown that $\frac{\varphi
(P)}{P}\geq\frac{e^{-\gamma}-\epsilon(P)}{\ln2}\frac{1}{\ln\ln N}$ , where
$\epsilon(P)$ is a monotone decreasing sequence converging to zero.

\begin{theorem}
$\lim\inf\frac{\varphi(N)}{N/\ln\ln N}=e^{-\gamma}$
\end{theorem}

where $\gamma=0.57721566$ is Euler's constant and where $e^{-\gamma
}=0.5614594836$.

This may cause us to repeat the experiment $\Omega(\frac{1}{\ln\ln N})$ times
in order to get $(d,P)=1$.

We note that we needed to add a condition on the period $P$ that $P^{2}\leq N
$ \ or $P\leq\sqrt{N}$ in order for the proof of the corollary to work.

\subsection{Testing if $P_{1}=P$ when $s$ is known or is $0$}

We can easily test if $s=0$ by checking to see if $f(0)=1.$

Now given a putative value of the period $P_{1}$ and a known offset or shift
$s$, how can we test whether $P_{1}=P$ ?

Assuming we have access to the Oracle to test individual values, we can
confirm $f(s)=1$ since $s$ is known. We will show that if $f(s+P_{1})=1$ and
$f(s+(M-1)P_{1})=1$ then $P_{1}=P.$

Case 1: If $P_{1}>P$ then $s+(M-1)P_{1}>s+(M-1)P.$ But $s+(M-1)P$ is the
largest index $x$ such that $f(x)=1.$ Therefore if $P_{1}>P$ we must have
$f(s+(M-1)P_{1})=0.$

Case 2: If $0<P_{1}<P$ then $s<s+P_{1}<s+P$ but between $s$ and $P$ there are
no other values $x$ such that $f(x)=1.$Therefore if $0<P_{1}<P$ we must have
$f(s+P_{1})=0.$

Therefore if $f(s)=1,f(s+P_{1})=1$ and $f(s+(M-1)P_{1})=1$ we must have
$P_{1}=P.$

\subsection{Testing if $(s_{1},P_{1})=(s,P)$ when $s$ is from a small known
set and $s\neq0$}

If we assume $s$ is unknown and $s\neq0$ but is from a small known set of
possible values such that we can exhaust over this set on a classical computer
and we are given a putative value of the period $P_{1}$, how can we test
whether a pair of values $(s_{1},P_{1})$ is the correct pair $(s,P)$ ?

We need only test whether $f(s_{1})=1$, $f(s_{1}+P_{1})=1$ and $f(s_{1}%
+(M-1)P_{1})=1$ where M is assumed known.

Case 1: If $s_{1}<s$ then $f(s_{1})=0$ since $\ s$ is the smallest index $x$
with $\ f(x)=1.$

Case 2: If $s_{1}>s$ and $f(s_{1})=1$ then $s_{1}=s+rP$ with $r>0$ . If
$f(s_{1}+P_{1})=1$ then $s_{1}+P_{1}=s+tP=s_{1}+(t-r)P$ with $t>r>0.$ Hence
$P_{1}=(t-r)P>0.$ If $f(s_{1}+(M-1)P_{1})=1$ then $s_{1}+(M-1)P_{1}%
=s+rP+(M-1)(t-r)P>s+(M-1)P$ which is the largest index $x$ with $f(x)=1.$
Therefore $f(s_{1}+(M-1)P_{1})=0.$

Hence if $f(s_{1})=1$, $f(s_{1}+P_{1})=1$ and $f(s_{1}+(M-1)P_{1})=1$ we must
have $s_{1}=s$ and then by following the case when $s$ is known we must also
have $P_{1}=P.$

Therefore if one or more of the values $f(s_{1}),$ $f(s_{1}+P_{1}),$
$f(s_{1}+(M-1)P_{1})$ is zero, either $s_{1}$ or $P_{1}$ is wrong. For a given
$P_{1\text{ }}$we must exhaust over all possible values of $s$ before we can
be sure that $P_{1}\neq P.$ For in the case that $P_{1}\neq P,$ we will have
for every possible $s_{1}$ that at least one of the values $f(s_{1}),$
$f(s_{1}+P_{1}),$ $f(s_{1}+(M-1)P_{1})$ is zero. In such a case we must try
another putative $P_{1}.$

\subsection{Finding $s\neq0$ using a Quantum Computer}

We can assume $s\neq0$ as the case $s=0$ is trivial and was considered above.
Let $s=\alpha+\beta P$ where $\alpha=s\operatorname{mod}P$ so that
$0\leq\alpha\leq P-1$ and $0\leq\alpha+\beta P+(M-1)P\leq N-1.$

We assume we are given the correct value of $P_{.}$ If $P$ is wrong, it will
be detected in the algorithm.

Step 1:

We create an initial superposition on $N$ values%

\[
|\psi_{1}>=\frac{1}{\sqrt{N}}\sum_{x=0}^{N-1}|x>
\]

and apply the Oracle $f$ and put this into the amplitude. We then apply Grover
without measurement to amplify the amplitudes and we have the following state%

\[
|\psi_{1}>=a_{k}\sum_{\times\in A}|x>+b_{k}\sum_{x\notin A}|x>
\]

where%

\[
a_{k}=\frac{1}{\sqrt{M}}\sin(2k+1)\theta,b_{k}=\frac{1}{\sqrt{N-M}}%
\cos(2k+1)\theta
\]

are the appropriate amplitudes of the states and where
\[
\sin\theta=\sqrt{M/N},\cos\theta=\sqrt{1-M/N}%
\]

Next we measure the register and with probability exceeding $1-M/N$ we will
measure a value $x_{1}\in A$ where $x_{1}=s+r_{1}P$ with $0\leq r_{1}\leq
M-1.$ Note that the total probability of the set A is given by
\begin{align*}
\Pr(x  &  \in A)=M(\frac{1}{\sqrt{M}}\sin(2k+1)\theta)^{2}=\sin^{2}%
(2k+1)\theta=1-\cos^{2}(2k+1)\theta\\
&  \Longrightarrow\Pr(x\in A)=\sin^{2}(2k+1)\theta\geq1-\frac{M}{N}%
\end{align*}

Now using our measured value $x_{1}=s+r_{1}P$ with $0\leq r_{1}\leq M-1$ we
check that $f(x_{1})=1$ and $f(x_{1}-P)=1.$ If $f(x_{1}-P)=0$ then either the
value of $P$ we are using is wrong or we have $r_{1}=0$ and $x_{1}=s.$ If we
test $f(s)=1$, $f(s+P)=1$ and $f(s+(M-1)P)=1$ then we have the correct $P$ and
$s$ otherwise $P$ is wrong. So assuming $f(x_{1}-P)=1$ we must have either the
correct $P$ or a multiple of $P$. We can use the procedure in Step 2 or Step
2' to find $s.$ The method in Step 2 uses the Exact Quantum Counting algorithm
to find $s$ (See ref[11] for details). The method in Step 2' uses a method of
decreasing sequence of measurements to find $s.$

Step 2 (using the Exact Quantum Counting algorithm):

Let $T$ be such that $T\geq M$ is the smallest power of $2$ greater than $M$.
We form a superposition%

\[
|\varphi_{1}>=\frac{1}{\sqrt{T}}\sum_{x=0}^{T-1}|x>|0>
\]

\noindent and apply the function $g(x)=Max(0,x_{1}-(x+1)P)$ where
$x_{1}=s+r_{1}P$ is our measured value, with $0\leq r_{1}\leq M-1$and put the
values of $g(x)$ into the second register to get%

\[
|\varphi_{2}>=\frac{1}{\sqrt{T}}\sum_{x=0}^{T-1}|x>|g(x)>
\]

\noindent Notice that as $x$ increases from $0$, $g(x)$ is a decreasing
sequence $s+rP$ with $r=(r_{1}-x-1).$ When $g(x)$ dips below $0$ we set
$g(x)=0$ to ensure $g(x)\geq0.$ Now we apply $f$ to $g(x)$ and put the results
into the amplitude to get%

\[
|\varphi_{3}>=\frac{1}{\sqrt{T}}\sum_{x=0}^{T-1}(-1)^{f(g(x))}|x>|g(x)>
\]

Notice that $f(g(x))=1$ when $s\leq g(x)<s+r_{1}P$ and is $0$ elsewhere. We
apply the exact quantum counting algorithm which determines how many values
$f(g(x))=1.$Let this total be $R.$ If $P$ is correct we expect $R=r_{1}$ and
we can determine $s=x_{1}-RP=s+r_{1}P-RP.$ We can then test if we have the
correct pair of values $s,P$ by testing whether $f(s)=1$, $f(s+P)=1$ and
$f(s+(M-1)P)=1.$ If this test fails then $P$ must be an incorrect value and we
must repeat the period finding algorithm.

We use Theorem 8.3.4 of ref[11]: The Exact Quantum Counting algorithm requires
an expected number of applications of $U_{f}$ in $O(\sqrt{(R+1)(T-R+1)}$ and
outputs the correct value $R$ with probability at least $2/3.$

Step 2' (decreasing sequence of measurements method):

Let $T$ be such that $T\geq M$ is the smallest power of $2$ greater than $M$.
We form a superposition%

\[
|\varphi_{1}>=\frac{1}{\sqrt{T}}\sum_{x=0}^{T-1}|x>|0>
\]

and apply the function $g(x)=Max(0,x_{1}-(x+1)P)$ where $x_{1}=s+r_{1}P$ with
$0\leq r_{1}\leq M-1$and put these values into the second register to get%

\[
|\varphi_{2}>=\frac{1}{\sqrt{T}}\sum_{x=0}^{T-1}|x>|g(x)>
\]

Notice that as $x$ increases from $0$, $g(x)$ is a decreasing sequence $s+rP$
with $r=(r_{1}-x-1).$ When $g(x)$ dips below $0$ we set $g(x)=0$ to ensure
$g(x)\geq0.$ Now we apply $f$ to $g(x)$ and put the results into the third
register and then into the amplitude.%

\[
|\varphi_{3}>=\frac{1}{\sqrt{T}}\sum_{x=0}^{T-1}(-1)^{f(g(x))}|x>|g(x)>
\]

Notice that $f(g(x))=1$ when $s\leq g(x)<s+r_{1}P$ and is $0$ elsewhere.

We then run Grover without measurement to amplify the amplitudes and measure
the second register containing $g(x).$

With probability close to 1 we will measure a new value $x_{2}=s+r_{2}P$ with
$0\leq r_{2}<r_{1}.$ We test the values $f(x_{2})=1$ and $f(x_{2}-P)=1.$ If
$f(x_{2}-P)=0$ then either the value of $P$ we are using is wrong or we have
$r_{2}=0$ and $x_{2}=s.$ If we test $f(s)=1$, $f(s+P)=1$ and $f(s+(M-1)P)=1$
then we have the correct $P$ and $s$ otherwise $P$ is wrong. So assuming
$f(x_{2}-P)=1$ we must have either the correct $P$ or a multiple of $P$. We
repeat this algorithm and go to Step 2' replacing the value $x_{1} $ in the
function $g(x)$ with $x_{2}$ etc. As we repeat the algorithm we will measure a
decreasing sequence of values $x_{1},x_{2}...$ that converges to $s.$ This
procedure will eventually terminate with the correct pair of values $P$ and
$s$ or we will determine that we have been using an incorrect value of $P$ and
we must repeat the quantum algorithm for finding putative $P $ and repeat the process.

How many times do we expect to repeat Step 2'? When we make our first
measurement we expect $r_{1}=M/2.$ For our second measurement we expect
$r_{2}=r_{1}/2$ etc. Therefore we expect to repeat this algorithm $O($
$\ln_{2}(M))$ times.

\section{\textbf{Replacing the QFT With a General Unitary Transform U}}

In general, if we had any Oracle $f$ which is $1$ on a set of labels $A$ and
$0$ elsewhere and we replaced the QFT\ with any unitary transform $U$ which
performs the following%

\[
|z>\rightarrow\frac{1}{\sqrt{N}}\sum_{y=0}^{N-1}\alpha(y,z)|y>
\]

\noindent we can compute the $AmpRatio(y)\ =\frac{Amplified-Amplitude(U)}%
{Amplitude(U)}$as follows.

\noindent As before, we have the following state after applying $U_{f}$:%

\[
|\psi>=\frac{1}{\sqrt{N}}\left[  (-2)\sum_{z\in A}|z>+\sum_{z=0}%
^{N-1}|z>\right]
\]

\noindent Next we apply the general unitary transform $U$ to obtain the state%

\[
U|\psi>=\sum_{y=0}^{N-1}\left[  \frac{(-2)}{N}\sum_{z\in A}\alpha
(y,z)+\frac{1}{N}\sum_{z=0}^{N-1}\alpha(y,z)\right]  |y>
\]

\noindent In the Amplified-U case we apply Grover without measurement followed
by $U$ we obtain the state%

\[
|\phi_{k}>=\sum_{y=0}^{N-1}\left[  \frac{(a_{k}-b_{k})}{\sqrt{N}}\sum_{z\in
A}\alpha(y,z)+\frac{b_{k}}{\sqrt{N}}\sum_{z=0}^{N-1}\alpha(y,z)\right]  |y>
\]

\noindent If $\sum_{z=0}^{N-1}\alpha(y,z)=0$ and $\sum_{z\in A}\alpha
(y,z)\neq0$ we get the same $AmpRatio(y)$ formula that we obtained when
$U=QFT$%

\begin{align*}
AmpRatio(y)  &  =\frac{\ \frac{(a_{k}-b_{k})}{\sqrt{N}}\sum_{z\in A}%
\alpha(y,z)+\frac{b_{k}}{\sqrt{N}}\sum_{z=0}^{N-1}\alpha(y,z)}{\frac{(-2)}%
{N}\sum_{z\in A}\alpha(y,z)+\frac{1}{N}\sum_{z=0}^{N-1}\alpha(y,z)}\\
&  =\frac{\ \frac{(a_{k}-b_{k})}{\sqrt{N}}\sum_{z\in A}\alpha(y,z)}%
{\frac{(-2)}{N}\sum_{z\in A}\alpha(y,z)}\\
&  =\frac{\ \frac{(a_{k}-b_{k})}{\sqrt{N}}}{\frac{(-2)}{N}}\\
&  =\frac{(a_{k}-b_{k})}{-2}\sqrt{N}\\
&  =\frac{1}{-2}\left[  \sqrt{\frac{N}{M}}\sin(2k+1)\theta-\sqrt{\frac{N}%
{N-M}}\cos(2k+1)\theta\right] \\
&  =\frac{N}{-2M}\tan\theta\sin2k\theta
\end{align*}

\noindent This gives
\[
\Pr Ratio(y)=\frac{N^{2}}{4M^{2}}\tan^{2}\theta\sin^{2}2k\theta
\]

\bigskip As in the case when U=QFT, we have the following inequality for the
$\Pr Ratio(y)$ for a general U, the increase in the probability due to
amplification when $k=\left\lfloor \frac{\pi}{4\theta}\right\rfloor $ and
making use of $\frac{M}{N}(\frac{N}{N-M})\geq\tan^{2}\theta\sin^{2}%
2k\theta\geq\frac{M}{N}(\frac{N}{N-M})(\frac{N-2M}{N})^{2}$%

\begin{align*}
\frac{N}{4M}(\frac{N}{N-M})  &  \geq\Pr Ratio(y)\geq\frac{N}{4M}(\frac{N}%
{N-M})(1-\frac{2M}{N})^{2}\\
&  \Longrightarrow\Pr Ratio(y)\approx\frac{N}{4M}%
\end{align*}

\textbf{References}

[1] Nakahara and Ohmi, \textquotedblleft Quantum Computing: From Linear
Algebra to Physical Realizations\textquotedblright, CRC Press (2008).

[2] S. Lomonaco, \textquotedblleft Shor's Quantum Factoring
Algorithm,\textquotedblright\ AMS PSAPM, vol. 58, (2002), 161-179.

[3] P. Shor, \textquotedblleft Polynomial time algorithms for prime
factorization and discrete logarithms on a quantum computer\textquotedblright,
SIAM J. on Computing, 26(5) (1997) pp1484-1509 (quant-ph/9508027).

[4] L. Grover, \textquotedblleft A fast quantum mechanical search algorithm
for database search\textquotedblright, Proceedings of the 28th Annual ACM
Symposium on Theory of Computing (STOC 1996), (1996) 212-219.

[5] Hardy and Wright \textquotedblleft An Introduction to the Theory of
Numbers\textquotedblright, Oxford Press Fifth Edition (1979).

[6] S. Lomonaco and L. Kauffman, \textquotedblleft Quantum Hidden Subgroup
Algorithms: A Mathematical Perspective,\textquotedblright\ AMS CONM, vol. 305,
(2002), 139-202.

[7] S. Lomonaco, \textquotedblleft Grover's Quantum Search
Algorithm,\textquotedblright\ AMS PSAPM, vol. 58, (2002), 181-192.

[8] S. Lomonaco and L. Kauffman, \textquotedblleft Is Grover's Algorithm a
Quantum Hidden Subgroup Algorithm?,\textquotedblright\ Journal of Quantum
Information Processing, Vol. 6, No. 6, (2007), 461-476.

[9] G. Brassard, P. Hoyer, M. Mosca and A. Tapp, "Quantum Amplitude
Amplification and Estimation", AMS CONM, vol 305, (2002), 53-74.

[10] M. Nielsen and I. Chuang, "Quantum Computation and Quantum Information",
Cambridge University Press (2000).

[11] P. Kaye, R. Laflamme and M. Mosca, "An Introduction to Quantum
Computing", Oxford University Press (2007).

[12] N. Yanofsky and M. Mannucci, "Quantum Computing For Computer Scientists",
Cambridge University Press (2008).

[13] S. Lomonaco, "A Lecture on Shor's Quantum Factoring Algorithm Version
1.1",quant-ph/0010034v1 9 Oct 2000.

\end{document}